\documentclass[12pt,a4paper]{article}
\usepackage{amsmath}
\usepackage{amssymb}
\usepackage{amsthm}
\usepackage{amsmath}
\usepackage{amssymb}
\usepackage{authblk}
\usepackage[hidelinks]{hyperref}
\usepackage{bbm}
\usepackage{url}
\usepackage{amscd}
\usepackage{mathrsfs}
\usepackage{graphicx}
\usepackage{float}
\usepackage[margin=0.7in]{geometry}\newtheorem{theorem}{Theorem}
\newtheorem{corollary}[theorem]{Corollary}
\newtheorem{lemma}[theorem]{Lemma}

\newtheorem*{remark}{Remark}

\newcommand{\ep}{\varepsilon}
\newcommand{\E}{\mathbb{E}}

\newcommand{\R}{\mathcal{R}}
\newcommand{\C}{\mathcal{C}}

\newcommand{\F}{\mathbb{F}}
\renewcommand{\P}{\mathbb{P}}

\renewcommand{\l}{\ell}

\title{Decoding Error Probability of the Random Code Ensemble over the Erasure Channel} 
 \author[1,2]{Chin Hei Chan}
 
 \affil[1*]{\small{Hetao Institute of Mathematics and Interdisciplinary Sciences, Shenzhen 518017, Guangdong, China}}
 
\affil[2]{\small{Hong Kong University of Science and Technonlogy, Clear Water Bay, Hong Kong, China}}

\affil[ ]{\href{mailto:chenzhanxi@himis-sz.cn}{\texttt{chenzhanxi@himis-sz.cn}}}
\date{}
\begin{document}
\maketitle
\begin{abstract}
In this paper, we provide explicit formulas for the average decoding error probabilities of the ensemble of all $(n,M)_q$ codes over the erasure channel under unambiguous, list and maximum-likelihood decoding principles. Moreover, we derive the asymptotic error exponents for the ensemble of all $(n,q^{Rn})_q$ codes.
\end{abstract}
\begin{keywords}
Decoding error probability, error exponent, random code ensemble, erasure channel, unambiguous decoding, list decoding, maximum likelihood decoding.
\end{keywords}
\section{Introduction}
\subsection{Background}
The binary erasure channel, first introduced by Elias \cite{Elias}, is a discrete memoryless channel in which each information bit is either correctly transmitted or totally erased. While the binary version is widely used in practice, it is natural in a theoretical aspect to extend to the $q$-ary erasure channel, that is, the information input symbols come from the finite field $\F_q$ of $q$ elements.

The decoding problem under the erasure channel has attracted increasing interest in the past few decades, owing to its wide application in the Internet as well as in distributed storage systems to resist random packet losses \cite{Byers,Luby,Lun}. Among the several decoding principles, three are of particular interest: the unambiguous decoding, list decoding, and maximum likelihood decoding. The decoding error probabilities and asymptotic error exponents of particular linear codes or linear code ensembles under these three decoding principles have been extensively investigated (see \cite{Didier,Lemes,SF,Weber,Xiong} and references therein).

In \cite{SF} Shen and Fu computed the decoding error probabilities of some linear codes, and more generally, the average decoding error probabilities of the random linear code ensemble under the three decoding principles by introducing the notion of incorrigible sets. They showed that the incorrigible set distributions are related to the support weight distributions of linear codes. From this, they also derived the asymptotic error exponents of the ensemble for fixed rate under unambiguous decoding as the code length tends to infinity. Xiong \cite{Xiong} further simplified their formulas for the decoding error probabilities, and computed the asymptotic error exponents under fixed-size list and maximum likelihood decoding as well. In \cite{CFX} the authors did the analogous problem for a similar but different code ensemble known as the random parity-check matrix ensemble, by showing that the average incorrigible set distributions are related to the rank distribution of the matrices, a quantity that is simpler than the support weight distributions of linear codes.

In this paper we investigate the decoding performance of the $q$-ary erasure channel under the three decoding principles for general codes, including nonlinear ones. Unlike linear codes which allow one to characterize the incorrigible sets through rank conditions, nonlinear codes lack such algebraic structure. Moreover, they do not have translation symmetry around codewords, which means the incorrigible set distribution depends on the choice of the center codeword as well. This makes computation of such distribution for nonlinear codes intractable.

To circumvent this difficulty, we instead consider the random ensemble of all $(n,M)_q$ codes. Our main contribution is to derive explicit formulas for the average decoding error probabilities of this ensemble under the three different decoding principles, as well as the asymptotic error exponents as code length increases to infinity with fixed rate.

\subsection{Statement of Main Results}

The first result is the average decoding error probabilities of the ensemble of all $(n,M)_q$ codes under the three decoding principles.
\begin{theorem}\label{Exp}
Let $\C$ denote the random ensemble of all $(n,M)_q$ codes and $\ep$ be the probability that a symbol is erased upon transmission through the $q$-ary erasure channel.
\begin{enumerate}
\item The average unsuccessful decoding probability $P_{\mathrm{ud}}(\C,\ep)$ of $\C$ under unambiguous decoding is given by
\begin{equation}\label{Pud}
P_{\mathrm{ud}}(\C,\ep)=\sum_{i=1}^n\left[1-\frac{\binom{q^n-q^i}{M-1}}{\binom{q^n-1}{M-1}}\right]\binom{n}{i}\ep^i(1-\ep)^{n-i};
\end{equation}
\item The average unsuccessful decoding probability $P_{\mathrm{ld}}(\C,L,\ep)$ of $\C$ under list decoding with list size $L$ is given by
\begin{equation}\label{Pld}
P_{\mathrm{ld}}(\C,L,\ep)=\sum_{i=1}^n\sum_{j=L+1}^M\frac{\binom{q^i-1}{j-1}\binom{q^n-q^i}{M-j}}{\binom{q^n-1}{M-1}}\binom{n}{i}\ep^i(1-\ep)^{n-i};
\end{equation}
\item The average decoding error probability $P_{\mathrm{mld}}(\C,\ep)$ of $\C$ under maximum likelihood decoding is given by
\begin{equation}\label{Pmld}
P_{\mathrm{mld}}(\C,\ep)=\sum_{i=1}^n\sum_{L=2}^M\frac{\binom{q^i-1}{L-1}\binom{q^n-q^i}{M-L}(1-L^{-1})}{\binom{q^n-1}{M-1}}\binom{n}{i}\ep^i(1-\ep)^{n-i}.
\end{equation}
\end{enumerate}
\end{theorem}
These formulas look quite similar with those for the $[n,k]_q$ linear code ensemble in \cite[Theorem 2]{SF} with usual binomial coefficients in place of Gaussian binomial coefficients and size of the summation index set for list size much larger (exponential rather than linear in $n$).

Now let $M=q^{Rn}$ for some $R \in (0,1)$. Here $R=n^{-1}\log_q M$ is the \emph{rate} of the code. The error exponents (in $q$-its) of the random ensemble of all $(n,q^{Rn})_q$ codes under the three decoding principles are given as follows.
\begin{theorem}\label{Errexp}
Let $\C_{n,Rn}$ be the random ensemble of all $(n,q^{Rn})_q$ codes, where the rate $R \in (0,1)$ is fixed and $n \to \infty$, and $\ep$ be the probability that a symbol is erased upon transmission through the $q$-ary erasure channel.
\begin{enumerate}
\item The error exponents $T_{\mathrm{ud}}(\ep)$ and $T_{\mathrm{mld}}(\ep)$ for the average unsuccessful decoding probability of $\C_{n,Rn}$ under unambiguous decoding and maximum likelihood decoding are both given by
\begin{equation}\label{Errud}
T_{\mathrm{ud}}(\ep)=T_{\mathrm{mld}}(\ep)=\begin{cases}
0 &(1-\ep \leq R < 1)\\
(1-R)\log_q\left(\frac{1-R}{\ep}\right)+R\log_q\left(\frac{R}{1-\ep}\right) &\left(\frac{1-\ep}{1-\ep+q\ep} < R < 1-\ep\right)\\
1-R-\log_q(1-\ep+q\ep) &\left(0 < R \leq \frac{1-\ep}{1-\ep+q\ep}\right).
\end{cases}
\end{equation}
\item The error exponent $T_{\mathrm{ld}}(L,\ep)$ for the average unsuccessful decoding probability of $\C_{n,Rn}$ under list decoding with fixed list size $L$ is given by
\begin{equation}\label{Errld}
T_{\mathrm{ld}}(L,\ep)=\begin{cases}
0 &(1-\ep \leq R < 1)\\
(1-R)\log_q\left(\frac{1-R}{\ep}\right)+R\log_q\left(\frac{R}{1-\ep}\right) &\left(\frac{1-\ep}{1-\ep+q^L\ep} < R < 1-\ep\right)\\
L(1-R)-\log_q(1-\ep+q^L\ep) &\left(0 < R \leq \frac{1-\ep}{1-\ep+q^L\ep}\right).
\end{cases}
\end{equation}
\item The error exponent $\widetilde{T}_{\mathrm{ld}}(\lambda,\ep)$ for the average unsuccessful decoding probability of $\C_{n,Rn}$ under list decoding with list size $L=q^{\lambda n}$ ($\lambda < R$) is given by
\begin{equation}\label{Errlde}
\widetilde{T}_{\mathrm{ld}}(\lambda,\ep)=\begin{cases}
0 &(\min\{1-\ep+\lambda,1\}\leq R < 1)\\
(1-R+\lambda)\log_q\left(\frac{1-R+\lambda}{\ep}\right)+(R-\lambda)\log_q\left(\frac{R-\lambda}{1-\ep}\right) &(\lambda < R < \min\{1-\ep+\lambda,  1\}).
\end{cases}
\end{equation}
\end{enumerate}
\end{theorem}
The above formulas show that the code can be decoded efficiently in average under any of the three decoding principles if the code rate $R$ (or, in the exponential list-size regime, the adjusted rate $R-\lambda$) is less than the capacity $1-\ep$ of the erasure channel, which is consistent with Shannon's channel coding theorem \cite{Gallager, SGB}.

Statement 3 of Theorem \ref{Errexp} only deals with $\lambda < R$. Indeed, if $\lambda \geq R$, then the code size itself does not exceed the list size, and so no decoding error can occur at all, which renders study of the error exponent (which is always negative infinity) meaningless.

Moreover, we note that the formulas for the error exponents under unambiguous and maximum likelihood decoding in (\ref{Errud}) are exactly the same as those for the random linear code ensemble $\mathscr{C}_{n,Rn}$ (\cite[Theorem 3]{SF} and \cite[Theorem 1.3]{Xiong}) and random parity-check matrix ensemble $\R_{(1-R)n,n}$ (\cite[Eq. (10)]{CFX}). On the other hand, the formula for the error exponent under fixed-size list decoding in (\ref{Errld}) is in the same form as those for $\mathscr{C}_{n,Rn}$ (\cite[Theorem 1.4]{Xiong}) and $\R_{(1-R)n,n}$ (\cite[Eq. (9)]{CFX}) with $L$ in place of $\l+1$. The list size for those two ensembles is $q^\l$, which is larger than $\l+1$ whenever $q^\l > 2$. Hence the list decoding performance with respect to $\mathcal{C}_{n,Rn}$ is superior to $\mathscr{C}_{n,Rn}$ and $\R_{(1-R)n,n}$ for small rates. To illustrate this more clearly, we have generated a plot (see Figure \ref{Error}) to compare the list decoding error exponents of $\mathcal{C}_{n,Rn}$ and $\mathscr{C}_{n,Rn}$ (or $\mathcal{R}_{(1-R)n,n}$) with list size 4 over the binary erasure channel ($q=2$), with erasure probability $\ep=0.25$.
\begin{figure}[htb!]
\includegraphics[angle=0,width=1.0 \textwidth,height=0.4\textheight]{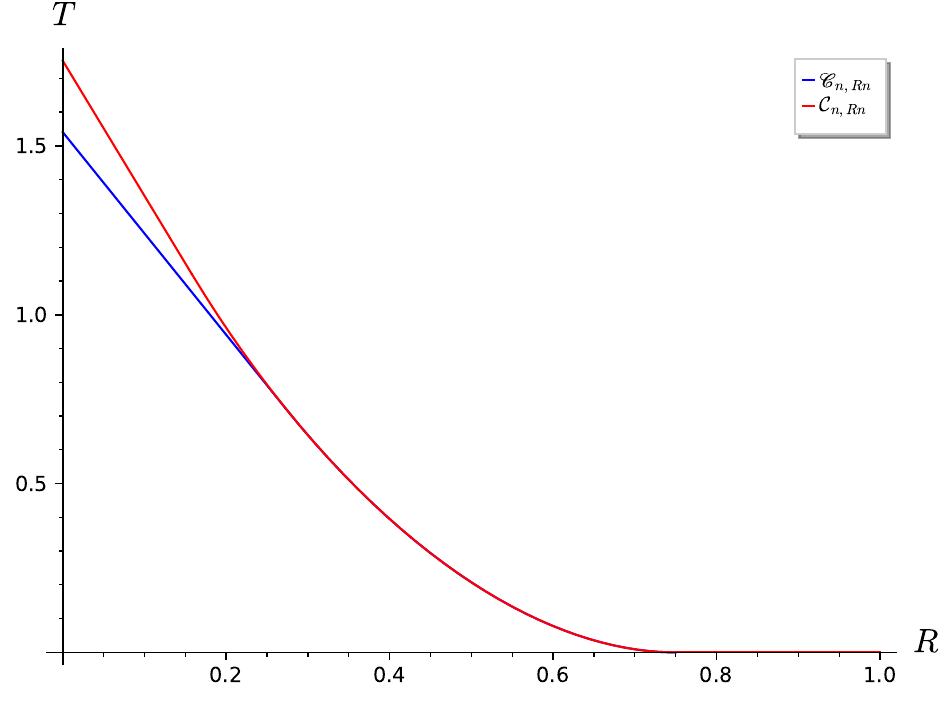}
\caption{$4$-list decoding error exponent of random linear code ensemble $\mathscr{C}_{n,Rn}$ VS random (nonlinear) code ensemble $\mathcal{C}_{n,Rn}$ over the binary erasure channel, with erasure probability $\ep=0.25$.}\label{Error}
\end{figure}


This paper is organized as follows. In Section \ref{Pre}, we describe the three decoding principles in more details, and derive expressions for the decoding error probabilities of a general (nonlinear) code under the three decoding principles. Then in Sections \ref{Pthm1} and \ref{Pthm2} we prove Theorems \ref{Exp} and \ref{Errexp} respectively.
\section{Preliminaries}\label{Pre}
\subsection{Three Decoding Principles}
In this section we recall the three decoding principles of the $q$-ary erasure channel. We first introduce some notations.

First, under the $q$-ary erasure channel, the channel input alphabet is the finite field $\mathbb{F}_q$ of order $q$. Upon transmission, the symbol is either received correctly with probability $1-\ep$ or erased with probability $\ep$, where $0 < \ep < 1$.

For an integer $1 \leq M \leq q^n$, an $(n,M)_q$ code $C$ (not necessarily linear) is simply an $M$-subset of the vector space $\F_q^n$. Write explicitly $C=\{\mathbf{x}_1,\mathbf{x}_2,\cdots,\mathbf{x}_M\}$.

Denote $\mathcal{Y}:=\F_q \cup \{\square\}$, where $\square$ is the erasure symbol, and $[n]:=\{1,2,\cdots,n\}$.

For any $\mathbf{r} \in \mathcal{Y}^n$, define the sets
$$E(\mathbf{r}):=\{\alpha \in [n]: r_\alpha=\square\}$$
and
$$C(\mathbf{r}):=\{\mathbf{c} \in C: c_\alpha=r_\alpha \quad\forall \alpha \in [n] - E(\mathbf{r})\}.$$
Here $C(\mathbf{r})$ is simply the set of all codewords that can result in receiving $\mathbf{r}$, that is, the complete output list upon decoding $\mathbf{r}$ if there is no size limit. Hence for each fixed $\mathbf{c} \in C$ and $\mathbf{r} \in \mathcal{Y}^n$, we have
$$\P(\mathbf{r}\text{ is received}|\mathbf{c}\text{ is sent})=\begin{cases}
\ep^{\#E(\mathbf{r})}(1-\ep)^{n-\#E(\mathbf{r})} &(\mathbf{c} \in C(\mathbf{r}))\\
0 &(\text{otherwise}),
\end{cases}$$
where for any finite set $A$, $\# A$ denotes the cardinality of $A$.

In addition, for any subset $E \subseteq [n]$ and codeword $m$ (identifying with $\mathbf{x}_m$), we define
$$C(E,m):=\{\mathbf{x}_{m'} \in C: x_{m'\alpha}=x_{m\alpha} \ \forall \alpha \in [n]-E\}.$$
Note that $C(E,m) \neq \emptyset$ as it always contains $\mathbf{x}_m$.

The three decoding principles are as follows:
\begin{itemize}
\item In {\bf unambiguous decoding}, the decoder outputs the only one codeword in $C({\bf r})$ if $\#C({\bf r})=1$ and declares ``failure'' otherwise. The unsuccessful decoding probability $P_{\mathrm{ud}}(C,\ep)$ of $C$ under unambiguous decoding is defined to be the probability that the decoder declares ``failure''. 
    
\item In {\bf list decoding}, the decoder with list size $L$ outputs all the codewords in $C(\mathbf{r})$ if $\#C(\mathbf{r})\le L$ and declares ``failure'' otherwise. The unsuccessful decoding probability $P_{\mathrm{ld}}(C,L,\ep)$ of $C$ under list decoding with list size $L$ is defined to be the probability that the decoder declares ``failure''.
    
\item In {\bf maximum likelihood decoding}, the decoder randomly chooses a codeword in $C(\bf{r})$ uniformly and outputs this codeword. The decoding error probability $P_{\mathrm{mld}}(C,\ep)$ of $C$ under maximum likelihood decoding is defined to be the probability that the codeword outputted by the decoder is not the sent codeword. 
\end{itemize}

\subsection{Decoding Error Probability of a General Code}
We now figure out the decoding error probability of an $(n,M)_q$ code $C$ under the three decoding principles.

The unsuccessful decoding probability of codeword $m$ (that is, $\mathbf{x}_m$) under unambiguous decoding is
\begin{align*}
P_\mathrm{ud}(C,m,\ep)&=\sum_{\mathbf{r} \in \mathcal{Y}^n} \P(\#C(\mathbf{r}) > 1, \mathbf{r} \text{ received}|\mathbf{x}_m \text{ sent})\\
&=\sum_{E \subseteq [n]}\P(\#C(E,m) > 1, E|\mathbf{x}_m \text{ sent})\\
&=\sum_{\substack{E \subseteq [n] \\ \#C(E,m) > 1}}\P(E|\mathbf{x}_m \text{ sent})\\
&=\sum_{i=1}^n I_i(C,m)\ep^i(1-\ep)^{1-i},
\end{align*}
where
$$I_i(C,m):=\#\{E \subseteq [n]: \#E=i, \#C(E,m) > 1\}.$$
A set $E \subseteq [n]$ that satisfies the condition $\#C(E,m) > 1$ is called an \emph{incorrigible set of $C$ centered at $m$}. The collection $\{I_i(C,m)\}_{i=1}^n$ is then called the \emph{incorrigible set distribution of $C$ centered at $m$}.

Then the overall unsuccessful decoding probability of $C$ under unambiguous decoding is given by
\begin{align}\label{UD}
P_\mathrm{ud}(C,\ep)&=\frac{1}{M}\sum_{m=1}^M P_\mathrm{ud}(C,m,\ep)\nonumber\\
&=\sum_{i=1}^n I_i(C)\ep^i(1-\ep)^{1-i},\end{align}
where
$$I_i(C):=\frac{1}{M}\sum_{m=1}^M I_i(C,m).$$

The collection $\{I_i(C)\}_{i=1}^n$ is called the \emph{average incorrigible set distribution} of $C$.

Similarly, the unsuccessful decoding probability of codeword $m$ (that is, $\mathbf{x}_m$) under list decoding with list size $L$ is
\begin{align*}
P_\mathrm{ld}(C,L,m,\ep)&=\sum_{\mathbf{r} \in \mathcal{Y}^n} \P(\#C(\mathbf{r}) > L, \mathbf{r} \text{ received}|\mathbf{x}_m \text{ sent})\\
&=\sum_{E \subseteq [n]}\P(\#C(E,m) > L, E|\mathbf{x}_m \text{ sent})\\
&=\sum_{\substack{E \subseteq [n] \\ \#C(E,m) > L}}\P(E|\mathbf{x}_m \text{ sent})\\
&=\sum_{i=1}^n I_i^{[L]}(C,m)\ep^i(1-\ep)^{1-i},
\end{align*}
where
$$I_i^{[L]}(C,m):=\#\{E \subseteq [n]: \#E=i, \#C(E,m) > L\}.$$
A set $E \subseteq [n]$ that satisfies the condition $\#C(E,m) > L$ is called an \emph{$L$-incorrigible set of $C$ centered at $m$}. The collection $\left\{I_i^{[L]}(C,m)\right\}_{i=1}^n$ is then called the \emph{$L$-incorrigible set distribution of $C$ centered at $m$}.

Then the overall unsuccessful decoding probability of $C$ under list decoding is given by
\begin{align}\label{PLD}
P_\mathrm{ld}(C,L,\ep)&=\frac{1}{M}\sum_{m=1}^M P_\mathrm{ld}(C,L,m,\ep)\nonumber\\
&=\sum_{i=1}^n I_i^{[L]}(C)\ep^i(1-\ep)^{1-i},\end{align}
where
$$I_i^{[L]}(C):=\frac{1}{M}\sum_{m=1}^M I_i^{[L]}(C,m).$$

The collection $\{I_i^{[L]}(C)\}_{i=1}^n$ is called the \emph{average $L$-incorrigible set distribution} of $C$.

Note that
\begin{eqnarray}
I_i^{[1]}(C,m)&=&I_i(C,m),\label{I11}\\
I_i^{[1]}(C)&=&I_i(C),\nonumber\\
P_\mathrm{ld}(C,1,m,\ep)&=&P_\mathrm{ud}(C,m,\ep),\nonumber\\
P_\mathrm{ld}(C,1,\ep)&=&P_\mathrm{ud}(C,\ep)\nonumber.
\end{eqnarray}

Finally, for maximum likelihood decoding, we have
\begin{align*}
P_\mathrm{mld}(C,m,\ep)&=\sum_{\mathbf{r} \in \mathcal{Y}^n} \P(\mathbf{x}_{m'} \text{ decoded } \exists m' \neq m, \mathbf{r} \text{ received}|\mathbf{x}_m \text{ sent})\\
&=\sum_{L=2}^M(1-L^{-1})\sum_{\substack{E \subseteq [n] \\ \#C(E,m)=L}} \P(E|\mathbf{x}_m\text{ sent})\\
&=\sum_{i=1}^n\sum_{L=2}^M \lambda_i^{[L]}(C,m)(1-L^{-1})\ep^i(1-\ep)^{n-i}.
\end{align*}
Here
$$\lambda_i^{[L]}(C,m):=I_i^{[L-1]}(C,m)-I_i^{[L]}(C,m)=\#\{E \subseteq [n]: \#E=i, \#C(E,m)=L\}.$$

Therefore the overall decoding error probability is
\begin{equation}\label{MLD}
P_\mathrm{mld}(C,\ep)=\sum_{i=1}^n\sum_{L=2}^M \lambda_i^{[L]}(C)(1-L^{-1})\ep^i(1-\ep)^{n-i}
\end{equation}
where
$$\lambda_i^{[L]}(C):=\frac{1}{M}\sum_{m=1}^M \lambda_i^{[L]}(C,m)=I_i^{[L-1]}(C)-I_i^{[L]}(C).$$
Note that
\begin{eqnarray}
\sum_{j=L+1}^M \lambda_i^{[j]}(C,m)&=&I_i^{[L]}(C,m),\label{IL1}\\
\sum_{j=L+1}^M\lambda_i^{[j]}(C)&=&I_i^{[L]}(C).\nonumber
\end{eqnarray}
\begin{remark}
\begin{enumerate}
\item If $C$ is linear so that $M=q^k$ for some integer $k$, then the quantity $\#C(E,m)$ is independent of the choice of codeword $m$ (and in fact must be a power of $q$), and so are the quantities $I_i^{[L]}(C,m)$ for any $L$. In particular the latter is equal to $I_i^{[L]}(C)$. In this case we simply denote $I_i^{(\l)}(C):=I_i^{[q^\l]}(C)$. Together with $I_i(C)$, both are consistent with the definition in \cite[Section II]{SF}. In addition, we have, for $L \geq 2$,
$$\lambda_i^{[L]}(C)=\begin{cases}
I_i^{(\l-1)}(C)-I_i^{(\l)}(C) &(L=q^\l \ \exists \l \in \mathbb{N})\\
0 &(\text{otherwise}).
\end{cases}$$
Therefore our formulas are simply a generalization of those in \cite[Section II]{SF} to nonlinear codes.
\item In \cite{SF,Xiong}, the authors computed explicitly the incorrigible set distributions by relating to the support weight distributions. This is possible due to the fine algebraic structure of linear codes. However for nonlinear codes there seems no simple way to compute the incorrigible set distribution either centered at a single codeword or in average in general.
\end{enumerate}
\end{remark}
\section{Proof of Theorem \ref{Exp}}\label{Pthm1}
Now denote by $\C$ the random ensemble of choosing $M$ distinct elements $\mathbf{x}_1, \mathbf{x}_2, \cdots, \mathbf{x}_M$ from $\F_q^n$. Note that this ensemble takes each $(n,M)_q$-code with the same multiplicity $M!$, so it is essentially the same as the random $(n,M)_q$-code ensemble.

Denote by $\E$ the expectation over the ensemble $\C$.
\begin{lemma}\label{I}
For any $1 \leq i \leq n, 1 \leq m \leq M$ and $1 \leq L \leq M-1$,
\begin{enumerate}
\item $\E[\lambda_i^{[L]}(C,m)]=\frac{\binom{q^i-1}{L-1}\binom{q^n-q^i}{M-L}}{\binom{q^n-1}{M-1}}\binom{n}{i}$;
\item $\E[I_i^{[L]}(C,m)]=\sum_{j=L+1}^M \frac{\binom{q^i-1}{j-1}\binom{q^n-q^i}{M-j}}{\binom{q^n-1}{M-1}}\binom{n}{i}$;
\item $\E[I_i(C,m)]=\left[1-\frac{\binom{q^n-q^i}{M-1}}{\binom{q^n-1}{M-1}}\right]\binom{n}{i}$.
\end{enumerate}
\end{lemma}
\begin{proof}
\begin{enumerate}
\item First, we expand
\begin{equation}\label{I2}
\E[\lambda_i^{[L]}(C,m)]=\sum_{\substack{E \subseteq [n]\\ \#E=i}} \E[\mathbbm{1}_{\#C(E,m) =L}]=\sum_{\substack{E \subseteq [n] \\ \#E=i}} \P(\#C(E,m) = L).
\end{equation}
Now we look at the condition $\#C(E,m)=L$. First the codeword $m$ can be any word in $\F_q^n$. Given this, we need exactly $L-1$ other codewords to have coordinates coincide with $\mathbf{x}_m$ outside $E$, and they can be put in any $L-1$ positions of the code apart from $m$. This gives $\binom{M-1}{L-1}[q^i-1]_{L-1}[q^n-q^i]_{M-L}$ possibilities, where $[N]_t:=\frac{N!}{(N-t)!}$. Overall, we get
$$\P(\#C(E,m)=L)=\frac{q^n\binom{M-1}{L-1}[q^i-1]_{L-1}[q^n-q^i]_{M-L}}{[q^n]_M}=\frac{\binom{q^i-1}{L-1}\binom{q^n-q^i}{M-L}}{\binom{q^n-1}{M-1}}$$
after simplification.

Plugging this into (\ref{I2}) gives the desired result.
\item This follows from combining Statement 1 and (\ref{IL1}).
\item This follows from Statement 2, (\ref{I11}) and the fact that
$$\sum_{j=1}^M\P(\#C(E,m)=j)=1.$$
\end{enumerate}
\end{proof}
Lemma \ref{I} shows that the quantities $\E[\lambda_i^{[L]}(C,m)], \E[I_i^{[L]}(C,m)]$ and $\E[I_i(C,m)]$ are all independent of $m$. Hence by averaging over $m$ with $1 \leq m \leq M$, we have the following result.
\begin{corollary}\label{I0}
\begin{enumerate}
\item $\E[\lambda_i^{[L]}(C)]=\frac{\binom{q^i-1}{L-1}\binom{q^n-q^i}{M-L}}{\binom{q^n-1}{M-1}}\binom{n}{i}$;
\item $\E[I_i^{[L]}(C)]=\sum_{j=L+1}^M \frac{\binom{q^i-1}{j-1}\binom{q^n-q^i}{M-j}}{\binom{q^n-1}{M-1}}\binom{n}{i}$;
\item $\E[I_i(C)]=\left[1-\frac{\binom{q^n-q^i}{M-1}}{\binom{q^n-1}{M-1}}\right]\binom{n}{i}$.
\end{enumerate}
\end{corollary}
Now we can prove Theorem \ref{Exp}.
\begin{proof}[Proof of Theorem \ref{Exp}]
\begin{enumerate}
\item \textbf{Unambiguous decoding.}

By (\ref{UD}) and Statement 3 of Corollary \ref{I0}, we have
\begin{align*}
P_\mathrm{ud}(\C,\ep)&=\sum_{i=1}^n\E[I_i(C)]\ep^i(1-\ep)^{n-i}\\
&=\sum_{i=1}^n\left[1-\frac{\binom{q^n-q^i}{M-1}}{\binom{q^n-1}{M-1}}\right]\binom{n}{i}\ep^i(1-\ep)^{n-i},
\end{align*}
giving Equation (\ref{Pud}).

\item \textbf{List decoding.}

By (\ref{PLD}) and Statement 2 of Corollary \ref{I0}, we get
\begin{align*}
P_\mathrm{ld}(\C,L,\ep)
&=\sum_{i=1}^n\E[I_i^{[L]}(C)]\ep^i(1-\ep)^{n-i}\\
&=\sum_{i=1}^n\sum_{j=L+1}^M\frac{\binom{q^i-1}{j-1}\binom{q^n-q^i}{M-j}}{\binom{q^n-1}{M-1}}\binom{n}{i}\ep^i(1-\ep)^{n-i},
\end{align*}
which is exactly Equation (\ref{Pld}).

\item \textbf{Maximum likelihood decoding.}

By (\ref{MLD}) and Statement 1 of Corollary \ref{I0}, we obtain
\begin{align*}
P_\mathrm{mld}(\C,\ep)&=\sum_{i=1}^n\sum_{L=2}^M \E[\lambda_i^{[L]}(C)](1-L^{-1})\ep^i(1-\ep)^{n-i}\\
&=\sum_{i=1}^n\sum_{L=2}^M\frac{\binom{q^i-1}{L-1}\binom{q^n-q^i}{M-L}(1-L^{-1})}{\binom{q^n-1}{M-1}}\binom{n}{i}\ep^i(1-\ep)^{n-i}.
\end{align*}
Hence Equation (\ref{Pmld}) follows.
\end{enumerate}
This completes the proof of Theorem \ref{Exp}.
\end{proof}
\section{Proof of Theorem \ref{Errexp}}\label{Pthm2}
Denote by $\C_n$ the random ensemble of choosing $M=q^{Rn}$ distinct elements $\mathbf{x}_1,\mathbf{x}_2,\cdots,\mathbf{x}_{q^{Rn}}$ from $\F_q^n$, where $R \in (0,1)$ is the rate. 

The error exponents of the average decoding error probabilities of $\C_n$ are defined by
\begin{eqnarray*}
T_\mathrm{ud}(\ep)&=&-\lim_{n \to \infty}\frac{1}{n}\log_q P_\mathrm{ud}(\C_n,\ep) \text{ for unambiguous decoding,}\\
T_\mathrm{mld}(\ep)&=&-\lim_{n \to \infty}\frac{1}{n}\log_q P_\mathrm{mld}(\C_n,\ep) \text{ for maximum likelihood decoding,}\\
T_\mathrm{ld}(L,\ep)&=&-\lim_{n \to \infty}\frac{1}{n}\log_q P_\mathrm{ld}(\C_n,L,\ep) \text{ for list decoding with fixed list size $L$ and }\\
\widetilde{T}_\mathrm{ld}(\lambda,\ep)&=&-\lim_{n \to \infty}\frac{1}{n}\log_q P_\mathrm{ld}(\C_n,q^{\lambda n},\ep) \text{ for list decoding with list size $L=q^{\lambda n}$},
\end{eqnarray*}
provided that the limit exists \cite{Gallager, Merhav, SGB, SF, Viterbi}.

In other words, $T_\mathrm{ud}(\ep), T_\mathrm{mld}(\ep), T_\mathrm{ld}(L,\ep)$ and $\widetilde{T}_\mathrm{ld}(\lambda,\ep)$ are non-negative numbers only depending on $q, R, \ep$ (and $L,\lambda$) such that
\begin{eqnarray*}
    P_\mathrm{ud}(\C_n,\ep)&=&q^{-n(T_\mathrm{ud}(\ep)+o(1))},\\
    P_\mathrm{mld}(\C_n,\ep)&=&q^{-n(T_\mathrm{mld}(\ep)+o(1))},\\
    P_\mathrm{ld}(\C_n,L,\ep)&=&q^{-n(T_\mathrm{ld}(L,\ep)+o(1))} \text{ for fixed }L,\\
    P_\mathrm{ld}(\C_n,q^{\lambda n},\ep)&=&q^{-n(\widetilde{T}_\mathrm{ld}(\lambda,\ep)+o(1))} \text{ for fixed }\lambda \in (0,R)
\end{eqnarray*}
respectively.

First, we write
$$A_{i,j}:=\frac{\binom{q^i-1}{j-1}\binom{q^n-q^i}{M-j}}{\binom{q^n-1}{M-1}}.$$
Then we can simplify Equations (\ref{Pud}), (\ref{Pld}) and (\ref{Pmld}) as
$$P_\mathrm{ud}(\C_n,\ep)=\sum_{i=1}^n(1-A_{i,1})\binom{n}{i}\ep^i(1-\ep)^{n-i},$$
$$P_\mathrm{ld}(\C_n,L,\ep)=\sum_{i=1}^n\sum_{j=L+1}^M A_{i,j}\binom{n}{i}\ep^i(1-\ep)^{n-i}$$
and
$$P_\mathrm{mld}(\C_n,\ep)=\sum_{i=1}^n\sum_{L=2}^M A_{i,L}(1-L^{-1})\binom{n}{i}\ep^i(1-\ep)^{n-i}$$
respectively.

First, noting that unambiguous decoding is simply list decoding with list size 1,
and since
$$\frac{1}{2}P_\mathrm{ud}(\C_n,\ep) \leq P_\mathrm{mld}(\C_n,\ep) \leq P_\mathrm{ud}(\C_n,\ep),$$
we can easily derive
$$T_\mathrm{mld}(\ep)=T_\mathrm{ud}(\ep)=T_\mathrm{ld}(1,\ep).$$
Hence we only need to prove explicitly the results for list decoding under both fixed and exponential size regimes, that is, Statements 2) and 3) of Theorem \ref{Errexp}.

By the facts that
$$\max_{i=1}^N a_i \leq \sum_{i=1}^N a_i \leq N\max_{i=1}^N a_i$$
for any finite sequence $(a_i)_{i=1}^N$ of non-negative real numbers and
$$\lim_{n \to \infty} \frac{\log_q n}{n}=0,$$ we obtain
\begin{align}\label{Errld2}
\lim_{n \to \infty}\frac{1}{n}\log_q P_\mathrm{ld}(\C_n,L,\ep)&=\lim_{n \to \infty}\frac{1}{n}\log_q\max_{i=1}^n \left[\left(\sum_{j=L+1}^M A_{i,j}\right)\binom{n}{i}\ep^i(1-\ep)^{n-i}\right]\nonumber\\
&=\sup_{0 < t \leq 1} [F(t)+h(t)+t\log_q\ep+(1-t)\log_q(1-\ep)],
\end{align}
where $t:=\frac{i}{n}, h(t):=-t\log_q t-(1-t)\log_q(1-t)$ is the binary entropy function in $q$-its (see \cite{MacWilliams} and \cite[Lemma 3]{SF}), and
$$F(t):=\lim_{n \to \infty}\frac{1}{n}\log_q\left(\sum_{j=L+1}^M A_{tn,j}\right).$$

We start  by evaluating $A_{i,1}$.

Note that $R \in (0,1)$, so that $M-1=o(q^n-q^i)$ (as long as $i \neq n$) and $M-1=o(q^n-1)$.

Therefore we can write (see \cite{Spen})
$$\binom{q^n-q^i}{M-1} = \left(\frac{(q^n-q^i)e}{M-1}\right)^{M-1}(2\pi(M-1))^{-1/2}\exp\left(-\frac{(M-1)^2}{2(q^n-q^i)}(1+o(1))\right)$$
and similarly
$$\binom{q^n-1}{M-1} = \left(\frac{(q^n-1)e}{M-1}\right)^{M-1}(2\pi(M-1))^{-1/2}\exp\left(-\frac{(M-1)^2}{2(q^n-1)}(1+o(1))\right).$$
This implies
\begin{align*}
A_{i,1}&= \left(\frac{q^n-q^i}{q^n-1}\right)^{M-1}\exp\left[-\frac{(M-1)^2}{2}\left(\frac{1}{q^n-q^i}-\frac{1}{q^n-1}\right)(1+o(1))\right]\\
&=\left(1-\frac{q^i-1}{q^n-1}\right)^{M-1}\exp\left[-\frac{(M-1)^2}{2}\frac{q^i-1}{(q^n-q^i)(q^n-1)}(1+o(1))\right].
\end{align*}
Taking natural logarithm on both sides, we get
\begin{align*}
\ln A_{i,1}&=(M-1)\ln\left(1-\frac{q^i-1}{q^n-1}\right)-\frac{(M-1)^2}{2}\frac{q^i-1}{(q^n-q^i)(q^n-1)}(1+o(1))\\
&=-(M-1)\left[\frac{q^i-1}{q^n-1}+O\left(\left(\frac{q^i-1}{q^n-1}\right)^2\right)\right]-\frac{(M-1)^2}{2}\frac{q^i-1}{(q^n-q^i)(q^n-1)}(1+o(1))\\
&=-q^{i-n}M(1+o(1)).
\end{align*}
Therefore
$$A_{i,1}=\exp[-q^{i-n}M(1+o(1))].$$
Now it is easy to see that, for any $2 \leq j \leq \min\{q^i,M\}$,
\begin{equation}\label{Aratio}
\frac{A_{i,j}}{A_{i,j-1}}=\frac{(q^i-j+1)(M-j+1)}{(j-1)(q^n-q^i-M+j)}=\frac{q^{i-n}M}{j-1}\left(1-\frac{j-1}{q^i}\right)\left(1-\frac{j-1}{M}\right)\left(1-\frac{q^i+M-j}{q^n}\right)^{-1}.
\end{equation}
\subsection{Fixed List-Size Regime}
We first consider the case that $L$ is fixed. We divide into four different cases:

\textbf{Case 0.} $t=1$

In this case the only value for $j$ making the summand nonzero is $M$. For large enough $n$, we have $M = q^{Rn}> L$ and $\sum_{j=L+1}^M A_{n,j}=A_{n,M}=1$.

\textbf{Case 1.} $1-R < t < 1$

In this case $q^{i-n}M=q^{(t-1+R)n} \to \infty$ and we have, by (\ref{Aratio}), for any fixed $j \geq 2$ (that is, independent of $n$),
$$\frac{A_{i,j}}{A_{i,j-1}} \sim \frac{q^{i-n}M}{j-1},$$
so that, by induction, we obtain
\begin{align*}
    \sum_{j=1}^L
A_{i,j} &\sim \sum_{j=1}^L \frac{(q^{i-n}M)^{j-1}}{(j-1)!}\exp[-q^{i-n}M(1+o(1))]\\
    &\leq \frac{L(q^{i-n}M)^{L-1}}{(L-1)!}\exp[-q^{i-n}M(1+o(1))]=o(1).
\end{align*}
Hence
$$\sum_{j=L+1}^M A_{i,j} \sim 1.$$

\textbf{Case 2.} $t=1-R$

In this case $q^{i-n}M=1$, so by (\ref{Aratio}), for any fixed $j \geq 2$,
$$\frac{A_{i,j}}{A_{i,j-1}} \sim \frac{1}{j-1},$$
and
$$\sum_{j=1}^L A_{i,j}\sim \sum_{j=1}^L \frac{1}{(j-1)!}\exp(-1+o(1)) \leq 1-\frac{1}{L!e}+o(1).$$
Hence
$$\sum_{j=L+1}^M A_{i,j} \asymp 1.$$

\textbf{Case 3.} $0 < t < 1-R$

In this case it is clear that
$$\frac{q^i+M-j}{q^n}=o(1).$$
In addition, we have $q^{i-n}M=q^{(t-1+R)n}=o(1)$. Hence by (\ref{Aratio}), for any $2 \leq j \leq \min\{q^i,M\}$,
$$\frac{A_{i,j}}{A_{i,j-1}} \leq \frac{q^{i-n}M}{j-1}(1+o(1))=o(1).$$

By induction, we have
$$A_{i,j} \sim \frac{(q^{i-n}M)^{j-1}}{(j-1)!}\exp[-q^{i-n}M(1+o(1))] \sim \frac{(q^{i-n}M)^{j-1}}{(j-1)!}$$
for fixed $j$.

These imply
$$\sum_{j=L+1}^M A_{i,j} \sim A_{i,L+1} \sim \frac{(q^{i-n}M)^L}{L!}.$$

Combining all cases,
$$F(t)=-\lim_{n \to \infty} \frac{L(n-\log_q M-i)_+}{n}=-L(1-R-t)_+.$$
Putting this into (\ref{Errld2}), we obtain
$$T_\mathrm{ld}(L,\ep)=-\lim_{n \to \infty} \frac{1}{n}\log_q P_\mathrm{ld}(\C_n,L,\ep)=-\sup_{0 < t \leq 1} f(t),$$
where
\begin{equation}\label{f}
f(t)=-L(1-R-t)_++h(t)+t\log_q\ep+(1-t)\log_q(1-\ep).
\end{equation}
Note that this function is in the same form as \cite[Equation (16)]{CFX}, with $L$ in place of $\l+1$.

Hence we can mimic their computation to conclude that
$$T_{\mathrm{ld}}(L,\ep)=\begin{cases}
0 &(1-\ep \leq R < 1)\\
(1-R)\log_q\left(\frac{1-R}{\ep}\right)+R\log_q\left(\frac{R}{1-\ep}\right) &\left(\frac{1-\ep}{1-\ep+q^L\ep} < R < 1-\ep\right)\\
L(1-R)-\log_q(1-\ep+q^L\ep) &\left(0 < R \leq \frac{1-\ep}{1-\ep+q^L\ep}\right).
\end{cases}$$
That is, Equation (\ref{Errld}) holds.
\subsection{Exponential List-Size Regime}
Now we look at the exponential size regime of list decoding, that is, $L=q^{\lambda n}$ for some $\lambda \in (0,R)$. 

We again divide into several cases.

\textbf{Case 0.} $t=1$

This is similar to the fixed list-size regime.

\textbf{Case 1.} $1-R+\lambda < t < 1$

In this case, we have $M-L=o(q^n-q^i)$ and $L-1=o(q^i-1)$, so that (see \cite{Spen})
$$\binom{q^n-q^i}{M-L} = \left(\frac{(q^n-q^i)e}{M-L}\right)^{M-L}(2\pi(M-L))^{-1/2}\exp\left(-\frac{(M-L)^2}{2(q^n-q^i)}(1+o(1))\right)$$
and similarly
$$\binom{q^i-1}{L-1} = \left(\frac{(q^i-1)e}{L-1}\right)^{L-1}(2\pi(L-1))^{-1/2}\exp\left(-\frac{(L-1)^2}{2(q^i-1)}(1+o(1))\right).$$
Hence
\begin{align}\label{AL}
    A_{i,L}&=\frac{(q^i-1)^{L-1}(q^n-q^i)^{M-L}}{(q^n-1)^{M-1}}\frac{(M-1)^{M-1}}{(L-1)^{L-1}(M-L)^{M-L}}\sqrt\frac{M-1}{2\pi(L-1)(M-L)}\times\nonumber\\
    &\quad \exp\left(\frac{(M-1)^2}{2(q^n-1)}(1+o(1))-\frac{(L-1)^2}{2(q^i-1)}(1+o(1))-\frac{(M-L)^2}{2(q^n-q^i)}(1+o(1))\right)\nonumber\\
    &=\left(\frac{q^i-1}{q^n-1}\right)^{L-1}\left(1-\frac{q^i-1}{q^n-1}\right)^{M-L}\left(\frac{M-1}{L-1}\right)^{L-1}\left(1+\frac{L-1}{M-L}\right)^{M-L+1/2}\sqrt{\frac{1}{2\pi(L-1)}}\times\nonumber\\
    &\quad\exp\left(\frac{-q^iM^2+2q^nLM-q^nL^2}{2q^{2n}}(1+o(1))-\frac{L^2}{2q^i}(1+o(1))\right)\nonumber\\
    &=q^{n(L-1)(t-1+R-\lambda)}\left(1-\frac{q^i-1}{q^n-1}\right)^{M-L}\exp\left(\frac{2q^{\lambda n}-q^{n(t+2R-2)}+2q^{n(\lambda+R-1)}-q^{n(2\lambda-t)}}{2}(1+o(1))\right).
\end{align}
Since $t > 1-R+\lambda$, we get
$$t+2R-2>\lambda+R-1>2\lambda-t.$$
Thereby (\ref{AL}) yields
\begin{align*}
\ln A_{i,L}&=n(L-1)(t-1+R-\lambda)+(M-L)\ln\left(1-\frac{q^i-1}{q^n-1}\right)+\left(q^{\lambda n}-\frac{q^{n(t+2R-2)}}{2}\right)(1+o(1))\\
&\sim nq^{\lambda n}(t-1+R-\lambda)-q^{n(t-1+R)}-\frac{q^{n(t+2R-2)}}{2}\\
&\sim -q^{n(t-1+R)}.
\end{align*}
That is,
$$A_{i,L}=\exp[-q^{n(t-1+R)}(1+o(1))].$$
Moreover, by (\ref{Aratio}), for any $2 \leq j \leq L$,

$$\frac{A_{i,j}}{A_{i,j-1}} \geq \frac{q^{i-n}M}{L-1}(1+o(1))=q^{n(t-1+R-\lambda)}(1+o(1)) \to \infty.$$
Hence
$$\sum_{j=1}^L A_{i,j} \leq LA_{i,L}(1+o(1))=\exp[n\lambda \ln q-q^{n(t-1+R)}(1+o(1))]=\exp[-q^{n(t-1+R)}(1+o(1))]=o(1),$$
and so
$$\sum_{j=L+1}^M A_{i,j} \sim 1.$$

\textbf{Case 2.} $t=1-R+\lambda$

In this case we have $t+2R-2=\lambda+R-1=2\lambda-t$ and $t-1+R=\lambda$, so that all apparent leading terms in (\ref{AL}) get canceled. To get the precise limit to $\lim_{n \to \infty} \frac{1}{n} \log_q A_{i,L}$ and hence $\lim_{n \to \infty} \frac{1}{n}\log_q \left(\sum_{i=L+1}^M A_{i,j}\right)$, we need to expand the lower order terms as well. Nevertheless, this value must be non-positive, and it happens that this single value does not affect the error exponent. Therefore we skip this for sake of simplicity.

\textbf{Case 3.} $\lambda < t < 1-R+\lambda$

Then $t+2R-2 < \lambda+R-1 < 2\lambda-t < \lambda$. Hence (\ref{AL}) yields
\begin{align*}
    \ln A_{i,L}&=n(L-1)(t-1+R-\lambda)+(M-L)\ln\left(1-\frac{q^i-1}{q^n-1}\right)+q^{\lambda n}(1+o(1))\\
&\sim -nq^{\lambda n}(1-R+\lambda-t)-q^{n(t-1+R)}+q^{\lambda n}\\
&\sim -nq^{\lambda n}(1-R+\lambda-t).
\end{align*}
That is,
$$A_{i,L}=\exp[-nq^{\lambda n}(1-R+\lambda-t)(1+o(1))].$$
Moreover, by (\ref{Aratio}), for any $L+1 \leq j \leq \min\{q^i,M\}$,
$$\frac{A_{i,j}}{A_{i,j-1}} \leq \frac{q^{i-n}M}{L}(1+o(1))=q^{n(t-1+R-\lambda)}(1+o(1))=o(1).$$
This implies
\begin{align*}
\sum_{j=L+1}^M A_{i,j} &\sim A_{i,L+1}\\
&\sim q^{n(t-1+R-\lambda)}\exp[-nq^{\lambda n}(1-R+\lambda-t)(1+o(1))]\\
&\sim \exp[-nq^{\lambda n}(1-R+\lambda-t)(1+o(1))].
\end{align*}

\textbf{Case 4.} $0 < t \leq \lambda$

In this case we have $A_{i,j}=0$ for $L+1 \leq j \leq M$.

Combining all cases,
$$F(t)\begin{cases}
    =0 &(1-R+\lambda < t \leq 1)\\
    \leq 0 &(t=1-R+\lambda)\\
    =-\infty &(0 < t < 1-R+\lambda).
\end{cases}$$
Putting this into (\ref{Errld2}), we obtain
\begin{align*}
\widetilde{T}_\mathrm{ld}(\lambda,\ep)&=-\lim_{n \to \infty}\frac{1}{n}\log_q P_\mathrm{ld}(\C_n,q^{\lambda n},\ep)\\
&=-\max\left\{F(1-R+\lambda)+g(1-R+\lambda),\sup_{1-R+\lambda < t \leq 1}g(t)\right\}\\
&=-\max_{1-R+\lambda \leq t \leq 1}g(t),
\end{align*}
where
\begin{equation}\label{g}
g(t)=h(t)+t\log_q\ep+(1-t)\log_q(1-\ep),
\end{equation}
since $g$ is continuous over $[1-R+\lambda,1]$.

Differentiating with respect to $t$, we have
$$g'(t)=\log_q\left(\frac{\ep}{t}\right)-\log_q\left(\frac{1-\ep}{1-t}\right).$$
The only critical point is $t_0=\ep$, and $g$ attains local maximum there.
\begin{itemize}
\item If $\lambda < R < \min\{1-\ep +\lambda,1\}$, then $g$ is maximized at $t=1-R+\lambda$, so
$$\widetilde{T}_\mathrm{ld}(\lambda,\ep)=-g(1-R+\lambda)=(1-R+\lambda)\log_q\left(\frac{1-R+\lambda}{\ep}\right)+(R-\lambda)\log_q\left(\frac{R-\lambda}{1-\ep}\right).$$
\item If $\min\{1-\ep +\lambda,1\} \leq R < 1$, then $g$ is maximized at $t=\ep$, so
$$\widetilde{T}_\mathrm{ld}(\lambda,\ep)=-g(\ep)=0.$$
\end{itemize}
This completes the proof of Theorem \ref{Errexp}.
\begin{remark}
    Let $P$ and $\tilde{P}$ be Bernoulli distributions with parameters $\ep$ and $t$ respectively. Then the relative entropy (in $q$-its) from $\tilde{P}$ to $P$ is 
    \begin{equation}\label{D}
    D(\tilde{P}||P)=t\log_q\frac{t}{\ep}+(1-t)\log_q\frac{1-t}{1-\ep}.
    \end{equation}
    Moreover, if $Q$ is the uniform distribution on $\F_q$, then the mutual information $\tilde{I}(X;Y)$ between $X \sim Q$ and $Y$ with $Y|X=x \sim \tilde{P}$ for all $x \in \F_q$ is the capacity of the $q$-ary erasure channel with erasure probability $t$, that is,
    \begin{equation}\label{tI}
        \tilde{I}(X;Y)=1-t.
    \end{equation}
   We also have
   \begin{equation}\label{D2}
       D(\tilde{P}_{Y|X}||P_{Y|X}|Q):=\sum_{x \in \F_q} Q(x)D(\tilde{P}_{Y|X=x}||P_{Y|X=x})=\sum_{x \in \F_q}\frac{1}{q}D(\tilde{P}||P)=D(\tilde{P}||P).
   \end{equation}
    Putting (\ref{D})--(\ref{D2}) into (\ref{f}) and (\ref{g}), we obtain
   $$T_\mathrm{ld}(L,\ep)=\min_{\tilde{P}_{Y|X}}[L(\tilde{I}(X;Y)-R)_++D(\tilde{P}_{Y|X}||P_{Y|X}|Q)]$$
    and
    $$\widetilde{T}_\mathrm{ld}(\lambda,\ep)=\min_{\{\tilde{P}_{Y|X}: \tilde{I}(X;Y)\leq R-\lambda\}} D(\tilde{P}_{Y|X}||P_{Y|X}|Q),$$
    which are exactly the random coding error exponents for fixed-size and exponential-size list decoding derived by Merhav \cite[Section III-B and Section III-C]{Merhav}. That is, Theorem \ref{Errexp} provides an independent, self-contained derivation of Merhav's exponents for the $q$-ary erasure channel. Yet our work has made one further step by providing formulas for the exact average decoding error probability in Theorem \ref{Exp}, which are not present in \cite{Merhav}.
\end{remark}

\section*{Acknowledgments}
This research was supported by the research start-up fund of Hetao Institute of Mathematics and Interdisciplinary Sciences for the sole author. The proofs of Theorem 1 and Statements 1-2 of Theorem 2 (and their associated lemmas) were completed when he was working in Hong Kong University of Science and Technology (HKUST), and he would like to acknowledge the financial support provided by Department of Mathematics of HKUST as well.
\bibliographystyle{IEEEtranS}
\bibliography{CCH-2026}

@article{CFX,
author={Chan, C. H. and Fu, F. and Xiong, M.},
title={Decoding error probability of random parity-check matrix ensemble over the erasure channel},
journal={Des. Codes Cryptogr.},
volume={93},
pages={51-77},
year={2025}
}

@book{Gallager,
author={Gallager, R. G.},
title={Information Theory and Reliable Communication},
address={New York, NY, USA},
publisher={Wiley},
year={1968}
}

@book{MacWilliams,
author={MacWilliams, F. J. and Sloane, N. J. A.},
title={The Theory of Error-Correcting Codes},
address={Amsterdam, the Netherlands},
series={North-Holland Mathematical Library},
volume={16},
publisher={North-Holland Publishing Company},
year={1981}
}

@article{Merhav,
author={Merhav, N.},
title={{List Decoding---Random Coding Exponents and Expurgated Exponents}},
journal={IEEE. Trans. Inform. Theory},
volume={60},
number={11},
pages={6749-6759},
year={2014}
}

@article{SF,
author={Shen, L. and Fu, F.},
title={{The Decoding Error Probability of Linear Codes over the Erasure Channel}},
journal={IEEE. Trans. Inform. Theory},
volume={65},
number={10},
pages={6194-6203},
year={2019}
}

@book{Spen,
author={Spencer, J. and Florescu, L.},
title={Asymptopia},
series={Student Mathematical Library},
volume={71},
address={Providence, RI, USA},
publisher={American Mathematical Society},
year={2014},
pages={66}
}

@book{Viterbi,
author={Viterbi, A. J. and Omura, J. K.},
title={Principles of Digital Communication and Coding},
address={New York, NY, USA},
publisher={McGraw-Hill},
year={1979}
}

@article{Xiong,
author={Xiong, M.},
title={Decoding error probability in the erasure channel and the $r$-th support weight distribution},
journal={Sci. Sin. Math.},
volume={51},
pages={1-14},
year={2021},
note={(in Chinese)}
}

@article{SGB,
author={Shannon, C. E. and Gallager, R. G. and Berlekamp, E. R.},
title={{Lower Bounds to Error Probability for Coding on Discrete Memoryless Channels. I}},
journal={Inf. Control},
volume={10},
number={1},
pages={65-103},
year={1967}
}

@inproceedings{Byers,
author={Byers, J. W. and Luby, M. and Mitzenmacher, M. and Rege, A.},
title={A digital fountain approach to reliable distribution of bulk data},
booktitle={Proc. ACM SIGCOMM Conf. Appli., Technol., Architectures, Protocols Comput. Commun.},
address={Vancouuver, BC, Canada},
year={1998},
pages={56-67}
}

@inproceedings{Luby,
author={Luby, M. and Mitzenmacher, M. and Shokrollahi, A.},
title={Practical loss-resilient codes},
booktitle={Proc. 29th Annual ACM Symp. Theory Comput.},
year={1997},
pages={150-159}
}

@article{Lun,
author={Lun, D. S. and M\'edard, M. and Koetter, R. and Effros, M.},
title={On coding for reliable communication over packet networks},
journal={Phys. Commun.},
volume={1},
number={1},
pages={3-20},
Year={2008}
}

@article{Weber,
author={Weber, J. H. and Abdel-Ghaffar, K. A. S.},
title={Results on parity-check matrices with optimal stopping and/or dead-end set enumerators},
journal={IEEE Trans. Inform. Theory},
volume={54},
number={3},
pages={1412-1418},
year={1991}
}

@article{Didier,
author={Didier, F.},
title={A new upper bound on the block error probability after decoding over the erasure channel},
journal={IEEE Trans. Inform. Theory},
volume={52},
number={10},
pages={4496-4503},
year={2006}
}

@inproceedings{Lemes,
author={Lemes, L. C. and Firer, M.},
title={Generalized weights and bounds for error probability over erasure channels},
booktitle={Proc. Inf. Theory Appl. Workshop (ITA)},
address={San Diego, CA, USA},
year={2014},
pages={14}
}

@inproceedings{Elias,
    author={Elias, P.},
    title={{Coding for Noisy Channels}},
    booktitle={I.R.E. Convention record part 4},
    year={1955},
    pages={37-46}
}
\end{document}